\documentclass[onecolumn, conference,letterpaper]{IEEEtran}

\usepackage[utf8]{inputenc} 
\usepackage[T1]{fontenc}
\usepackage{url}
\usepackage{ifthen}
\usepackage{cite}
\usepackage[cmex10]{amsmath}

\usepackage{booktabs}       
\usepackage{amsfonts}
\usepackage{hyperref}
\usepackage{nicefrac}       
\usepackage{microtype}      
\usepackage{xcolor}         

\usepackage{enumitem}
\usepackage{bm}
\usepackage{amsmath}
\usepackage{amsthm}
\usepackage{caption}
\usepackage{amssymb}
\usepackage{subcaption}
\usepackage{wrapfig}
\usepackage{import}
\usepackage{tcolorbox}

\newtheorem{theorem}{Theorem}
\newtheorem{remark}{Remark}

\newtheorem{lemma}{Lemma}
\newtheorem{corollary}{Corollary}

\def\enc{\mathbf{u}_\textrm{enc}}
\def\dec{\mathbf{u}_\textrm{dec}}
\def\encs{u_\textrm{enc}}
\def\decs{u_\textrm{dec}}

\def\enclamb{\lambda_\textrm{enc}}
\def\declamb{\lambda_\textrm{dec}}
\def\lenc{\mathcal{L}_\textrm{enc}}
\def\ldec{\mathcal{L}_\textrm{dec}}

\def\nprcc{\texttt{LeTCC}}
\def\func{\mathbf{f}}
\def\stset{\mathcal{F}}
\def\fhat{\mathbf{\hat{f}}}
\def\est{\fhat_{\bm{\alpha},\bm{\beta}}[\enc, \dec, \stset]}

\newcommand{\norm}[1]{\left\lVert#1\right\rVert}

\newcommand{\hil}[1]{\mathcal{H}^{#1}\left(\Omega\right)}

\newcommand{\cmr}[1]{\mathcal{C}^{#1}\left(\Omega\right)}
\newcommand{\cm}[2]{\mathcal{C}^{#1}\left(\Omega; \mathbb{R}^#2\right)}

\newcommand{\hilm}[2]{\mathcal{H}^{#1}\left(\Omega; \mathbb{R}^#2\right)}

\newcommand{\lp}[1]{{L}^{#1}\left(\Omega\right)}

\newcommand{\lec}[2]{\stackrel{\text{#1}}{#2}}

\begin{document}
\title{General Coded Computing in a Probabilistic Straggler Regime} 


\author{%
  \IEEEauthorblockN{Parsa Moradi}
  \IEEEauthorblockA{University of Minnesota, Twin Cities\\
                    Minneapolis, MN, USA\\
                    moradi@umn.edu}
  \and
  \IEEEauthorblockN{Mohammad Ali Maddah-Ali}
  \IEEEauthorblockA{University of Minnesota, Twin Cities\\
                    Minneapolis, MN, USA\\
                    maddah@umn.edu}
}
\maketitle


\begin{abstract}
Coded computing has demonstrated promising results in addressing straggler resiliency in distributed computing systems. However, most coded computing schemes are designed for exact computation, requiring the number of responding servers to exceed a certain recovery threshold. Additionally, these schemes are tailored for highly structured functions. Recently, new coded computing schemes for general computing functions, where exact computation is replaced with approximate computation, have emerged. In these schemes, the availability of additional results corresponds to more accurate estimation of computational tasks. This flexibility introduces new questions that need to be addressed.
This paper addresses the practically important scenario in the context of general coded computing, where each server may become a straggler with a probability \( p \), independently from others. We theoretically analyze the approximation error of two existing general coded computing schemes: Berrut Approximate Coded Computing (\texttt{BACC}) and Learning Theoretic Coded Computing (\(\nprcc\)). Under the probabilistic straggler configuration, we demonstrate that the average approximation error for \texttt{BACC} and \(\nprcc\) converge to zero with high probability with the rate of at least \(\mathcal{O}(\log^3_{\nicefrac{1}{p}}(N)\cdot{N^{-3}})\) and \(\mathcal{O}(\log^4_{\nicefrac{1}{p}}(N)\cdot{N^{-2}})\), respectively. This is perhaps  surprising, as earlier results does not indicate a convergence when the number of stragglers scales with the total number of servers $N$. However, in this case, despite the average number of stragglers being $Np$, the independence of servers in becoming stragglers allows the approximation error to converge to zero.
These theoretical results are validated through experiments on various computing functions, including deep neural networks.
\end{abstract}

\section{Introduction}\label{sec:introduction}
Coded computing has emerged as a powerful tool to improve the reliability and security of distributed computing systems, particularly, in presence of slow servers, known as \emph{stragglers}, which fail to deliver results within a reasonable time, and adversarial servers that maliciously or unintentionally alter their outputs arbitrarily~\cite{jahani2022berrut,yu2020straggler,karakus2017straggler,soleymani2022approxifer, yu2019lagrange, moradicoded}.

Initially, the primary objective of coded computing was to achieve exact recovery, often within the context of finite fields. Within this framework, researchers have developed highly efficient solutions, particularly for structured computations such as polynomials \cite{yu2019lagrange, fahim2019numerically, fahimappx} and matrix multiplications \cite{yu2020straggler, yu2017polynomial, yu2020entangled, opt-recovery, short}.
These methods often impose a strict recovery threshold: if the number of non-straggler servers exceeds this threshold, exact recovery is achievable; otherwise, the computation process fails entirely \cite{yu2019lagrange,yu2017polynomial}.

The limited scope of conventional coded computing, however, does not address the requirements of modern distributed machine learning applications. In ML applications, computation tasks often lack a specific structure. Additionally, computations are performed over real numbers, where approximate results are generally sufficient. This motivates the development of \emph{general coded computing}, which can handle a wide range of computations and exploit the possibility of \emph{approximate recovery} to reduce the required computing resources \cite{moradicoded, jahani2022berrut}. Unlike exact recovery, approximate recovery does not rely on a strict recovery threshold. Instead, the approximation error decreases as the number of non-straggling servers increases, making this approach more practical for real-world scenarios where exact computation is often infeasible or unnecessary
\cite{moradicoded,soleymani2022approxifer,jahani2022berrut,martinez2024privacy}. In particular, in two  schemes, namely Berrut Approximate Coded Computing (\texttt{BACC})~\cite{jahani2022berrut} and Learning Theoretic Coded Computing ($\nprcc$)~\cite{moradicoded}, 
the approximation errors are respectively bounded by $\mathcal{O}(\frac{S^4}{N^2})$  and $\mathcal{O}(\frac{S^3}{N^3})$, where $N$ is the number of the servers and $S$ is number of stragglers. 

In this context, a practical question arises: if each server has a probability $p$ of becoming a straggler, independent of the others,  rather than imposing a limit of at most $S$
straggler servers, does the approximation error still converge to zero? If it does, what is the rate of convergence? One would naively argue that in this case, the number of stragglers is in average $pN$. Since $pN$ scales with $N$, the results of \cite{jahani2022berrut} and \cite{moradicoded} does not guarantee any convergence of approximation error to zero. Therefore, maybe in this case, we do not have any convergence. 

In this paper, we address this question and show that the above understanding is not correct. We theoretically analyze the  approximation error in the presence of probabilistic stragglers for two general coded computing schemes of \texttt{BACC} and $\nprcc$. We demonstrate that the approximation errors of \texttt{BACC} and $\nprcc$ schemes converge to zero at rates of at least $\mathcal{O}(\log^3_{1/p}(N) N^{-3})$ and $\mathcal{O}(\log^4_{1/p}(N) N^{-2})$, respectively, with high probability. This result demonstrates that the independence of servers in becoming stragglers indeed contributes positively to the convergence of the error.  In the end, we validate our theoretical findings through experiments on various computing functions, including deep neural networks (see Fig.~\ref{fig:rate_all}).

The paper is organized as follows: Section~\ref{sec:prelim} provides the preliminaries, and Section~\ref{sec:prob_form} defines the problem formulation. Section~\ref{sec:main_res} presents the main theoretical results, with proof sketches detailed in Section~\ref{sec:proof}. Lastly, Section~\ref{sec:exp_res} discusses the experimental results.

{\bf Notations:} In this paper, bold uppercase letters (e.g., $\mathbf{A}$) are used to denote matrices, while bold lowercase letters (e.g., $\mathbf{x}$) represent vectors. Coded matrices and vectors are indicated with a tilde, such as $\tilde{\mathbf{x}}$ and $\tilde{\mathbf{A}}$. The set $\{1, 2, \ldots, n\}$ is denoted by $[n]$, and the size of a set $S$ is written as $|S|$. The $i$-th component of a vector-valued function $\mathbf{f}$ is expressed as $f_i(\cdot)$. The first, second, and $k$-th order derivatives of a scalar function $f$ are represented as $f'$, $f''$, and $f^{(k)}$, respectively.
The $\ell_2$-norm of a vector $\mathbf{x}$ is denoted by $\norm{\mathbf{x}}_2$, while the $L_p$-norm of a function $f(\cdot)$ over the interval $\Omega$ is written as $\norm{f}_{\lp{p}}$. The space $\hilm{2}{m}$ corresponds to the reproducing kernel Hilbert space (RKHS) of second-order Sobolev functions for vector-valued functions of dimension $m$ over $\Omega$. Additionally, $\cm{2}{m}$ represents the space of all vector-valued functions of dimension $m$ that with a second derivative over the interval $\Omega$. For one dimension, we simply use $\hil{2}$ and $\cmr{2}$.

\section{Preliminaries}\label{sec:prelim}
\subsection{Coded Computing for Straggler Resistance}
Consider a scenario with one master node and $N$ servers. The master node aims to use the servers to compute $\{\func(\mathbf{x}_k)\}_{k=1}^K$ in a distributed manner, where $\mathbf{x}_k \in \mathbb{R}^d$. Here, $\func: \mathbb{R}^d \to \mathbb{R}^m$ can be any function, ranging from a simple one-dimensional function to a sophisticated deep neural network, with $K, d, m \in \mathbb{N}$. However, some servers can potentially become \emph{stragglers}, not finishing their tasks within the required deadline. Therefore, assigning each data point to a single server would not be an effective approach.

To address this issue, the master node uses coding. It sends $N$ coded data points to each server using an encoder function. Each coded data point is a combination of the original data points. Each server then applies $\func(\cdot)$ to its received coded data, returning these \emph{coded results} to the master node. The objective for the master node is to use some decoder function to approximate $\fhat(\mathbf{x}_k) \approx \func(\mathbf{x}_k)$, even when some servers act as stragglers. The redundancy built into the coded data and the corresponding coded results allows the master node to reconstruct the intended outputs, $\{\func(\mathbf{x}_k)\}_{k=1}^K$, using the results of the non-stragglers.

More formally, a general coded computing framework can be described as a three-step process:
\begin{enumerate}
    \item \textbf{Encoding:} The master node uses an encoder function $\enc: \mathbb{R} \to \mathbb{R}^d$ to map a set of fixed, distinct, and ordered points $\alpha_1 < \alpha_2 < \dots < \alpha_K \in \mathbb{R}$ (\emph{encoder mapping points}) into data points $\{\mathbf{x}_k\}_{k=1}^K$ in $ \mathbb{R}^d$. This mapping can be either approximate, i.e., $\mathbf{x}_k \approx \enc(\cdot)$ \cite{moradicoded} or precise, i.e., $\mathbf{x}_k = \enc(\cdot)$ \cite{jahani2022berrut}. Then, the same encoder function $\enc(\cdot)$ is applied to another set of fixed, distinct, and ordered points $\beta_1 < \beta_2 < \dots < \beta_N \in \mathbb{R}$ (\emph{decoder mapping points}). The master node subsequently sends the encoded data points $\tilde{\mathbf{x}}_n = \enc(\beta_n)$ to each server $n$, where each $\tilde{\mathbf{x}}_n$ represents a combination of all initial points $\{\mathbf{x}_k\}_{k=1}^K$.

    \item \textbf{Computing:} Each server $n \in [N]$ computes $\func(\tilde{\mathbf{x}}_n) = \func(\enc(\beta_n))$ on their given input and sends the result back to the master node.
    
    \item \textbf{Decoding:} The master node collects the available results $\{\func(\tilde{\mathbf{x}}_v)\}_{v \in \stset}$ from the set of non-straggler servers $\stset$. It then employs a decoder function $\dec: \mathbb{R} \to \mathbb{R}^m$ to map the points $\{\beta_v\}_{v \in \stset}$ to the servers' results $\{\func(\tilde{\mathbf{x}}_v)\}_{v\in \stset} = \{\func(\enc(\beta_{v}))\}_{v\in \stset}$. Again, the decoder mapping can be either approximate \cite{moradicoded} or accurate \cite{jahani2022berrut}. Finally, the master node computes $\hat{\func}(\mathbf{x}_k) := \dec(\alpha_k)$ as an estimate of $\func(\mathbf{x}_k)$ for $k \in [K]$ leveraging the relationship that  $\dec(\alpha_k) \approx \func(\enc(\alpha_k)) \approx \func(\mathbf{x}_k)$.
\end{enumerate}

\subsection{Coded Computing Frameworks for General Computing Function}
The first coded computing framework designed for general high-dimensional computing functions was Berrut Approximation Coded Computing (\texttt{BACC}) introduced by \cite{jahani2022berrut}. In this work, Berrut's rational interpolation \cite{berrut1988rational} is employed for mapping during both the encoding and decoding stages. More specifically, the encoder and decoder functions are defined as follows:
\begin{gather}\label{eq:berrut_enc}
    \enc(z)=\sum_{i=1}^{K} \frac{\frac{(-1)^i}{\left(z-\alpha_i\right)}}{\sum_{j=1}^{K} \frac{(-1)^j}{\left(z-\alpha_j\right)}} \mathbf{x}_i, \\ \label{eq:berrut_dec}
    \dec(z)=\sum_{v \in \stset} \frac{\frac{(-1)^v}{\left(z-\beta_v\right)}}{\sum_{i \in \stset} \frac{(-1)^i}{\left(z-\beta_i\right)}} \func\left(\enc\left(\beta_v\right)\right).
\end{gather}
Berrut interpolant benefits from high stability compared to other rational interpolations, as its denominator does not have poles for any set of interpolation points \cite{berrut1988rational}. Additionally, both encoding and decoding mapping are precise, meaning that $\enc(\alpha_k) = \mathbf{x}_k$ and  $\dec(\beta_v) = \func(\enc(\beta_v))= \func(\tilde{\mathbf{x}}_v)$ for $k \in [K]$ and $v \in \stset$. 

The authors in \cite{moradicoded} propose a new scheme called $\nprcc$, based on learning theory. They define an end-to-end loss function for the entire process and leverage learning theory to design the encoding and decoding mappings.
\begin{align}\label{eq:letcc_enc}
    \enc(\cdot)&=\underset{\mathbf{u} \in \hilm{2}{d}}{\operatorname{argmin}} \frac{1}{|\stset|} \sum^K_{k=1}\norm{\mathbf{u}\left(\alpha_k\right)-\mathbf{x}_k}^2_2 + \enclamb \sum^d_{j=1}  \int_{\Omega} \left(u_j''(t)\right)^2\,dt, 
    \\ \label{eq:letcc_dec}
    \dec(\cdot)&\underset{\mathbf{u} \in \hilm{2}{m}}{\operatorname{argmin}} \frac{1}{|\stset|} \sum_{v \in \stset}\norm{\mathbf{u}\left(\beta_v\right)-\func\left(\enc\left(\beta_v\right)\right)}^2_2 + \declamb \sum^m_{j=1}  \int_{\Omega} \left(u_j''(t)\right)^2\,dt,
\end{align}
where parameters $\declamb$ and $\enclamb$ are hyper-parameters that control the smoothness of the decoding and encoding functions, respectively.
\section{Problem Formulation}\label{sec:prob_form}
In contrast to previous problem settings in coded computing, which assume that at most $S$ servers are stragglers, this paper investigates a more practical scenario where each server has a probability $p$ of being a straggler.

Following the notation in~\cite{moradicoded}, let $\est(\cdot)$ denote the estimator function for the coded computing scheme. Let $\bm{\alpha} = [\alpha_1, \dots, \alpha_K]^T$ and $\bm{\beta} = [\beta_1, \dots, \beta_N]^T$. Furthermore, let $\mathcal{F} \sim F_N$ denote a random variable representing the set of non-straggler servers, where $F_N$ is a distribution over all subsets of $N$ servers. Consequently, the performance of a coded computing estimator $\fhat(\mathbf{x}) := \est(\cdot)$, applied to the data points $\{\mathbf{x_k}\}_{k=1}^K$, can be evaluated using the following average approximation error:
\begin{align}\label{eq:perf_metric}
    \mathcal{L}_\stset(\fhat) &:= \frac{1}{K} \sum^K_{k=1} \norm{\fhat(\mathbf{x}_k) - \func(\mathbf{x}_k)}^2_2
=\frac{1}{K}
    \sum^K_{k=1} \norm{\dec(\alpha_k)- \func(\mathbf{x}_k)}^2_2.
\end{align}
Our goal is to derive a high-probability upper bound for $\mathcal{L}_\stset(\fhat)$ and compare the performance of $\texttt{BACC}$ and $\nprcc$ schemes in probabilistic configuration of stragglers. 

\section{Main Results}\label{sec:main_res}
In this section, we provide our theoretical guarantees on the performance of the $\nprcc$ and $\texttt{BACC}$ frameworks under a probabilistic straggler model, where each server becomes a straggler with probability $p$. We establish an upper bound and analyze the convergence rate for both frameworks. For simplicity, we focus on a one-dimensional computing function $f:\mathbb{R} \to \mathbb{R}$. The results for both schemes are extendable to higher dimensions.

Let us define two variables representing minimum and maximum consecutive distance of decoder mapping points.
\begin{gather}\label{eq:max_dist_st}
\Delta_\textrm{max}:=\underset{n\in \{0,\dots,N\}}{\max} \{\beta_{n+1}-\beta_n\},\  \Delta_\textrm{min}:=\underset{n\in [N-1]}{\min} \left\{\beta_{n+1}-\beta_n\right\},
\end{gather}
Without loss of generality, we assume that the points $\{\alpha_k\}_{k=1}^K$ and $\{\beta_n\}_{n=1}^N$ lie within the interval $\Omega := (-1, 1)$.

The following theorems provide an upper bound for the $\nprcc$ and $\texttt{BACC}$ schemes respectively, under the probabilistic straggler configuration.
\begin{theorem}\label{th:letcc}
Consider the $\nprcc$ framework consisting of $N$ servers, where $\encs(\cdot), \decs(\cdot) \in \hil{2}$ and $\declamb \leqslant N^{-4}$. Let $f(\cdot)$ be a $\nu$-Lipschitz function satisfying $|f''| \leq \eta$ for some $\eta > 0$. Assume that each server becomes a straggler with probability $p := 1-q > 0$. If $\frac{\Delta_\textrm{max}}{\Delta_\textrm{min}} \leqslant B$ for some constant $B > 0$, then for any $\delta \in (0, 1)$, there exist constants $C_l$ and $n_0$ such that for $N > n_0$, with probability at least $1-\delta$, the following holds:
\begin{align}\label{eq:th_letcc}
    \mathcal{L}_\stset(\hat{f}) \leqslant C_l\frac{\left(\log_{\frac{1}{p}}(qN) + \sqrt{\frac{1}{\delta}}\right)^3}{N^3} \cdot \norm{(f \circ \encs)''}^2_{\lp{2}} +\frac{2\nu^2}{K} \sum_{k=1}^K \left(f(\encs(\alpha_k)) - f(x_k)\right)^2.
\end{align}

\end{theorem}

\begin{theorem}\label{th:berrut}
Consider the $\texttt{BACC}$ framework consisting of $N$ servers with $f \in \cmr{2}$. Assume that each server is a straggler with probability of $p := 1-q > 0$. If there is a constant $B$ such that $\frac{\Delta_\textrm{max}}{\Delta_\textrm{min}} \leqslant B$, then
for any $\delta \in (0, 1)$, there exist constants $C_b$ and $n_0$ such that for $N > n_0$ with a probability of at least  $1-\delta$, we have:
\begin{align}\label{eq:th_berrut}
    \mathcal{L}_\stset(\hat{f}){\leqslant} C_b\cdot B\frac{\left(\log_{\frac{1}{p}}(qN) + \sqrt{\frac{1}{\delta}}\right)^4}{N^2} \cdot \left(\sum^2_{i=1} \norm{(f \circ \encs)^{(i)}}_{\lp{\infty}}\right)^2.
\end{align}
\end{theorem}

\begin{corollary}\label{cor:rate}
By applying Theorems~\ref{th:letcc} and \ref{th:berrut}, in the probabilistic straggler setting, the average approximation error converges to zero with high probability. The convergence rates are at least $\mathcal{O}(\frac{\log^3_{1/p}(N)}{N^3})$ for $\nprcc$ and at least $\mathcal{O}(\frac{\log^4_{1/p}(N)}{N^2})$ for $\texttt{BACC}$.
\end{corollary}

\begin{remark}
{\normalfont Note that, in the probabilistic setting, the average number of stragglers in $N$ servers is approximately $Np$. Therefore, on average, the number of stragglers is a fraction of $N$. The upper bounds for the average approximation error of the $\texttt{BACC}$ and $\nprcc$ schemes with at most $S$ stragglers, as proved in \cite{moradicoded} and \cite{jahani2022berrut}, are $\mathcal{O}(\frac{S^3}{N^3})$ and $\mathcal{O}(\frac{S^4}{N^2})$, respectively. These results suggest that the average approximation error does not converge when the number of stragglers is a fraction of $N$. However, Corollary~\ref{cor:rate} demonstrates that the probabilistic behavior of stragglers causes the scheme's approximation error to converge to zero.}
\end{remark}
Some commonly used mapping points do not satisfy the conditions stated in Theorems~\ref{th:letcc} and \ref{th:berrut}. Notably, the first and second Chebyshev points—widely used for interpolation and adopted in coded computing frameworks \cite{jahani2022berrut}—are prominent examples. However, the following corollary demonstrates that a similar upper bound can still be established for Chebyshev points.
\begin{corollary}\label{cor:cheb}
    Consider the $\nprcc$ framework with $\declamb \leqslant N^{-6}$ and the $\texttt{BACC}$ framework, both consisting of $N$ servers. Assume the decoder mapping points are the first and second Chebyshev points, respectively: $\{\alpha_k\}_{k=1}^K = \cos(\frac{(2k-1)\pi}{2K})$ and $\{\beta_n\}_{n=1}^N = \cos(\frac{(n-1)\pi}{N-1})$. Under a probabilistic straggler configuration, the same upper bounds for convergence rates as stated in Corollary~\ref{cor:rate} can still be achieved.
\end{corollary}
\begin{figure*}[t]
     \centering
     \begin{subfigure}[b]{0.49\textwidth}
         \centering
         \includegraphics[width=0.9\textwidth, height=0.6\textwidth]{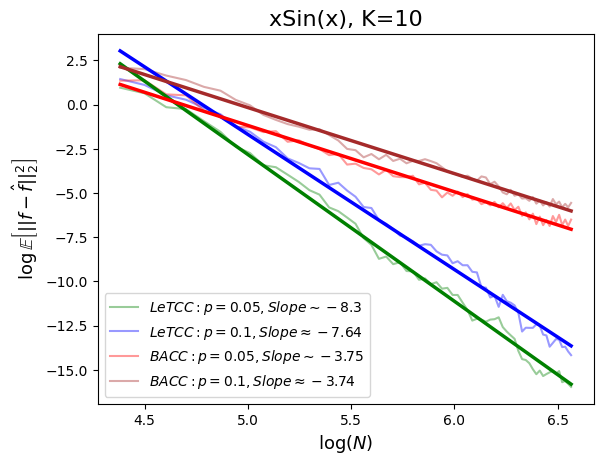}
         \label{fig:rate_xsinx}
     \end{subfigure}
     \hfill
     \begin{subfigure}[b]{0.49\textwidth}
         \centering
         \includegraphics[width=0.9\textwidth, height=0.6\textwidth]{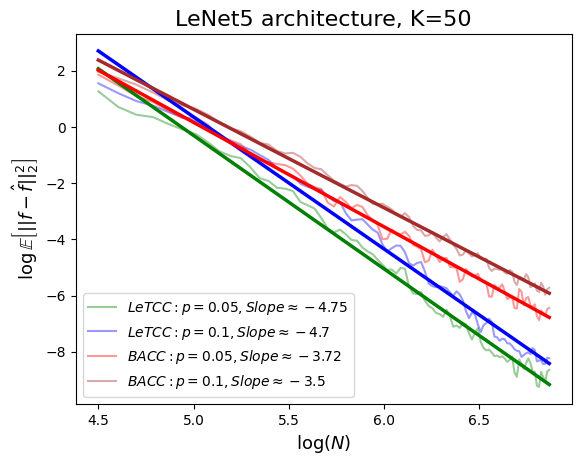}
         \label{fig:rate_lenet}
     \end{subfigure}
     \begin{subfigure}[b]{0.49\textwidth}
         \centering
         \includegraphics[width=0.9\textwidth, height=0.6\textwidth]{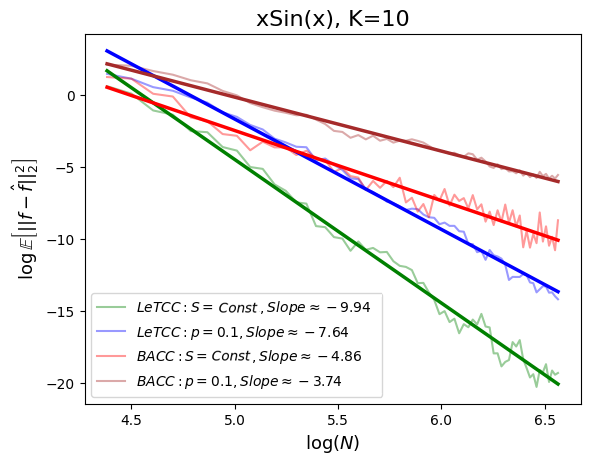}
         \label{fig:rate_xsinx_comp_with_cte}
     \end{subfigure}
     \hfill
     \begin{subfigure}[b]{0.49\textwidth}
         \centering
         \includegraphics[width=0.9\textwidth, height=0.6\textwidth]{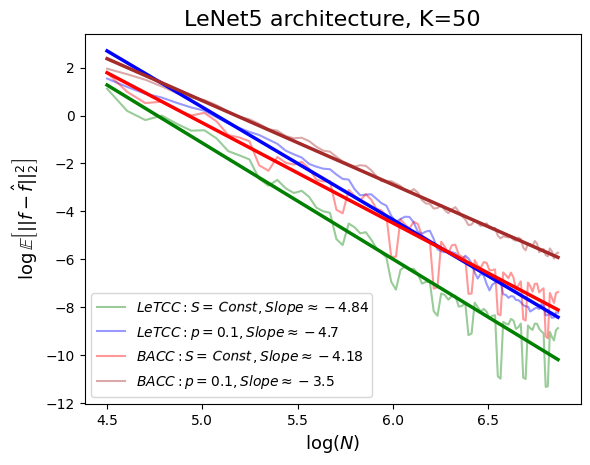}
         \label{fig:rate_lenet_comp_with_cte}
     \end{subfigure}
     \caption{Log-log plot showing the error convergence rates of $\nprcc$ and $\texttt{BACC}$ for the functions $f(x) = x\sin(x)$ and the LeNet5 network, with $p = 0.05$ and $p = 0.1$ (top). The plot also includes a comparison between $p = 0.1$ and the non-probabilistic configuration with $S = 2$ (bottom).}
     \label{fig:rate_all}
\end{figure*}
\section{Proof Sketch}\label{sec:proof}
In this section, we provide a proof sketch for the main theorems, with the formal proofs available in Appendix~\ref{app:proofs}. We first bound the average approximation error for each of  $\nprcc$ and $\texttt{BACC}$ schemes. 

{\bf (I) Bounding the average approximation error of  $\nprcc$:} Using the decomposition technique introduced in \cite{moradicoded}, we add and subtract $f(\encs(\alpha_k))$ and apply the triangle inequality to bound the approximation error: 
\begin{align}
    \label{eq:decompose}
   \mathcal{L}_\stset(\hat{f}) 
    \leqslant \underbrace{\frac{2}{K} \sum^K_{k=1} \left(\decs(\alpha_k) - f(\encs(\alpha_k))\right)^2}_{\ldec(\hat{f})} +  \underbrace{\frac{2}{K} \sum^K_{k=1} \left(f(\encs(\alpha_k)) - f(\mathbf{x}_k)\right)^2}_{\lenc(\hat{f})}.
\end{align}
Since $f(\cdot)$ is $\nu$-Lipschitz, we can derive an upper bound for $\lenc(\hat{f})$:
\begin{align}\label{eq:enc_bound}
    \lenc(\hat{f}) \leqslant \frac{2\nu^2}{K} \sum_{k=1}^K \left(f(\encs(\alpha_k)) - f(\mathbf{x}_k)\right)^2.
\end{align}
Next, we drive an upper bound for $\ldec(\hat{f})$. Defining $h(z) := \decs(z) - f(\encs(z))$, we obtain
$
    \ldec(\hat{f}) \leqslant 2\norm{h}^2_{\lp{\infty}}.
$
 Given that $f(\cdot)$ has bounded first and second derivatives, we can prove that $h \in \hil{2}$. Applying Sobolev interpolation inequalities \cite{lecun1998gradient}, we can derive an upper bound for $\norm{h}^2_{\lp{\infty}}$:
\begin{align*}
\norm{h}^2_{\lp{\infty}} 
  \leqslant 2 \left(\norm{h}^2_{\lp{2}} \cdot \norm{h'}^2_{\lp{2}}\right)^{\frac{1}{2}}.
\end{align*}
In the next step, we derive an upper bound for $\norm{h}^2_{\lp{2}}$ and $\norm{h'}^2_{\lp{2}}$.
Given the set of non-straggler servers $\stset$, let us define two variables representing the maximum and minimum distances between consecutive mapping points of the decoder function within the non-straggler set:
$
\Delta^\stset_\textrm{max}:=\underset{f \in \left\{0,\dots,|\stset|\right\}}{\max} \left\{\beta_{i_{f+1}}-\beta_{i_f}\right\}, \Delta^\stset_\textrm{min}:=\underset{f \in \left\{1,\dots,|\stset|-1\right\}}{\min} \left\{\beta_{i_{f+1}}-\beta_{i_f}\right\}.
$
The solution to \eqref{eq:letcc_dec} is a spline function, known as a \emph{second-order smoothing spline} \cite{wahba1975smoothing, wahba1990spline, wahba2006convergence}. The norm of the error and the norm of its first derivative are bounded as follows \cite{ragozin1983error}:
\begin{align}
    \norm{h}^2_{\lp{2}} {\leqslant}H_0\norm{(f\circ\encs)^{''}}^2_{\lp{2}}\cdot L,
\end{align}
and
\begin{align}
\norm{h'}^2_{\lp{2}} \leqslant H_1\norm{(f\circ\encs)^{''}}^2_{\lp{2}} \cdot L^{\frac{1}{2}}(1 + \frac{L}{16})^{\frac{1}{2}},
\end{align}
where $H_0$ and $H_1$ are constants. The term $L$ is given by $L = p_2\left(\frac{\Delta^\stset_\textrm{max}}{\Delta^\stset_\textrm{min}}\right) \cdot \frac{(N - |\stset|) \Delta^\stset_\textrm{max}}{4} \declamb + D \cdot \left(\Delta^\stset_\textrm{max}\right)^4$, where $D$ is a constant and $p_2(\cdot)$ is a second-order polynomial with constant and positive coefficients. 
If $\declamb \leqslant N^{-4}$, we can show that $L \leqslant \mathcal{O}((\Delta^\stset_\textrm{max})^4)$. Therefore, there exists a constant $C > 0$ such that:
\begin{align}\label{eq:bound_letcc}
\norm{h}^2_{\lp{\infty}} \leqslant C\cdot D \mathrm{E}_{\stset}\left[(\Delta^\stset_\textrm{max})^4\right]^\frac{1}{2} \cdot \sqrt{D}\mathrm{E}_{\stset}\left[(\Delta^\stset_\textrm{max})^2\right]^\frac{1}{2} \cdot \sqrt{H_0H_1} \norm{(f\circ\encs)^{''}}^2_{\lp{2}}.
\end{align}

{\bf (II) Bounding the average approximation error of  $\texttt{BACC}$:} 

In \cite{floater2007barycentric} it is shown that
$$
    \norm{h}_{\lp{\infty}} \leqslant \Delta^\stset_\textrm{max}(1+\mu) \cdot \sum^2_{i=1} \norm{(f \circ \encs)^{(i)}}_{\lp{\infty}},
$$
where $\mu := \max_{1 \leq i \leq |\stset|-2} \min \left\{\frac{\beta_{j_{i+1}}-\beta_{j_i}}{\beta_{j_i}-\beta_{j_{i-1}}}, \frac{\beta_{j_{i+1}}-\beta_{j_i}}{\beta_{j_{i+2}}-\beta_{j_{i+1}}}\right\}$ and $\stset := \{\beta_{j_1}, \dots, \beta_{j_{|\stset|}}\}$. Applying Max-Min inequality, we have:
$
    \mu \leqslant \min \left\{ \max_{1 \leq i \leq |\stset|-2} \frac{\beta_{j_{i+1}}-\beta_{j_i}}{\beta_{j_i}-\beta_{j_{i-1}}}, \max_{1 \leq i \leq |\stset|-2} \frac{\beta_{j_{i+1}}-\beta_{j_i}}{\beta_{j_{i+2}}-\beta_{j_{i+1}}}\right\} 
    \leqslant \frac{\Delta^\stset_\textrm{max}}{\Delta^\stset_\textrm{min}}.
$
Thus:
\begin{align}\label{eq:berrut_bound}
\mathrm{E}_{\stset}\left[\norm{h}^2_{\lp{\infty}}\right] \leqslant \mathrm{E}_{\stset}\left[\left(\Delta^\stset_\textrm{max}(1+\frac{\Delta^\stset_\textrm{max}}{\Delta^\stset_\textrm{min}})\right)^2\right]  \cdot\left(\sum^2_{i=1} \norm{(f \circ \encs)^{(i)}}_{\lp{\infty}}\right)^2.  
\end{align}

As a result, \eqref{eq:bound_letcc} and \eqref{eq:berrut_bound} show that the random variables $\Delta^\stset_\textrm{max}$ and $\Delta^\stset_\textrm{min}$ play a key role in determining the approximation error for both coded computing schemes. Let us define random variable $R_{\stset, N}$ as the maximum number of consecutive stragglers among the servers. Thus, we have:
\begin{align}\label{eq:longest_straggler_ineq}
     \Delta^\stset_\textrm{max} \leqslant (R_{\stset, N}+1) \cdot \Delta_\textrm{max}, \quad \Delta^\stset_\textrm{min} \geqslant \Delta_\textrm{min}.
\end{align}
Therefore, establishing an upper bound for $\Delta^\stset_\textrm{max}$ and $\frac{\Delta^\stset_\textrm{max}}{\Delta^\stset_\textrm{min}}$ is directly related to deriving an upper bound for $R_{\stset, N}$. Let us define i.i.d. Bernoulli random variables ${B_1, \dots, B_N}$, where $B_i=1$ indicates that the server $i$ is straggler. Thus, $\mathrm{E}[B_i] = p$ for $i \in [N]$. Consequently, $R_{\stset, N}$ represents the maximum length of consecutive ones in the sequence ${B_1, \dots, B_N}$. This random variable is commonly referred to as the \emph{longest head run} in the literature \cite{ni2019asymptotic,sinha2009distribution,bateman1948power,philippou1986successes,gordon1986extreme}.
\begin{lemma}(\cite[Theorem 2]{gordon1986extreme}) \label{lem:longest_run}Let $R_{\stset, N}$ be the maximal sub-sequence length of consecutive ones or length of longest run in an i.i.d. Bernoulli sequence $B_1,\dots, B_N$ with $\mathrm{E}[B_i] = p = 1 - q$ for $i \in [N]$. Let $\theta=\frac{\pi^2}{\ln \left(\frac{1}{p}\right)}$. Then:
$$
\left|\mathrm{E}\left[R_{\stset, N}\right]-\left(\log _{\frac{1}{p}}(qN)+\frac{\gamma}{\ln \left(\frac{1}{p}\right)}-\frac{1}{2}\right)\right|< g_1(\theta)+o(1),
$$
and
$$
\left|\operatorname{Var}\left[R_{\stset, N}\right]-\left(\frac{\pi^2}{6\left(\ln \left(\frac{1}{p}\right)\right)^2}+\frac{1}{12}\right)\right|< g_2(\theta)+o(1),
$$
where $\gamma$ is the Euler's constant, $g_1(\theta):=\frac{\theta^{\frac{1}{2}}}{2 \pi e^\theta(1-e^{-\theta})^2}$, and $g_2(\theta):=\frac{2(1.1+0.7 \theta) \theta^{\frac{1}{2}}}{2 \pi e^\theta(1-e^{-\theta})^3}$.
\end{lemma}
Lemma~\ref{lem:longest_run} states that although the $\mathrm{E}[R_{\stset, N}]$ grows with the $\log_{\frac{1}{p}}(N)$, its variance is nearly constant, i.e. $ g_2(\theta) - (\frac{\pi^2}{6(\ln (\frac{1}{p}))^2}+\frac{1}{12}) \leqslant \mathrm{Var[R_{\stset, N}]} \leqslant g_2(\theta) + (\frac{\pi^2}{6(\ln (\frac{1}{p}))^2}+\frac{1}{12})$.
This suggests that $R_{\stset, N}$ is highly predictable, with its value approximately equal to $\log_{\frac{1}{p}}(qN)$. Leveraging this observation, it can be shown that there exist constants $n_0$ and $C$ such that for $N > n_0$, the following inequality holds: 
\begin{align}\label{eq:run_prob} \Pr\left(R_{\stset, N} > C\left(\log_{\frac{1}{p}}(qN) + \sqrt{\frac{1}{\delta}}\right)\right) < \delta. 
\end{align} 
By combining \eqref{eq:run_prob} with the results in \eqref{eq:bound_letcc} and \eqref{eq:berrut_bound}, we can complete the proof.

\section{Experimental Results}\label{sec:exp_res}
In this section, we evaluate the performance of the $\nprcc$ and $\texttt{BACC}$ schemes under the probabilistic straggler configuration. We consider two classes of computing functions for our experiments. The first is a one-dimensional function, $f_1(x) = x\sin(x)$. The second is a deep neural network as a high dimensional function. Specifically, we use the LeNet5 architecture \cite{lecun1998gradient} trained on a handwritten image classification task. The network takes a $32 \times 32$ input image (flattened into a 1024-dimensional vector) and outputs a 10-dimensional vector representing the probability distribution over the classes. Thus, in this case, we are working with $\mathbf{f}_2:\mathbb{R}^{1024} \to \mathbb{R}^{10}$.

For the encoder and decoder mapping points, we use the first and second Chebyshev points, respectively as suggested by \cite{jahani2022berrut}. As an evaluation metric, we compute the empirical expectation of the recovery loss, \( \mathrm{E}_{\mathbf{x} \sim \mathcal{X}, \stset \sim F_{N}} \left[\frac{1}{K} \sum^K_{k=1} \norm{\dec(\alpha_k)- \func(\mathbf{x}_k)}^2_2\right] \). Specifically, for each value of $N$, we randomly sample input data points from the distribution and select stragglers with probability $p$. Each experiment is repeated $100$ times, and the average result is reported.

It is important to note that the smoothing parameters $\declamb$ and $\enclamb$ do not impact the convergence rate in noiseless computations, as long as $\declamb \leqslant N^{-4}$ (see \cite{moradicoded}). Therefore, in our experiments, we use the non-smooth version of $\nprcc$, setting $\declamb = \enclamb = 0$.

Our experimental results are consistent with the theoretical findings presented in Section~\ref{sec:main_res} and reveal the following key observations: for both class of computing functions (1) the $\nprcc$ scheme achieves a faster convergence rate compared to $\texttt{BACC}$, and (2) as $N$ increases, the convergence rate of average approximation error for the probabilistic configuration has a smaller exponent compared to the configuration with a maximum of $S$ stragglers (see Figure~\ref{fig:rate_all}).

\section*{Acknowledgment}
This material is based upon work supported by the National
Science Foundation under Grant CIF-2348638.



\onecolumn
\appendices
\section{Formal Proofs}\label{app:proofs}

\subsection{Proof of Theorem~\ref{th:letcc}}\label{sec:app_proof_th_letcc}
Using the decomposition technique introduced in \cite{moradicoded}, one can add and subtract $f(\encs(\alpha_k))$ to the objective \eqref{eq:perf_metric} and apply the triangle inequality to bound the average approximation error as follows:
\begin{align}
   \mathcal{L}_\stset(\hat{f}) 
    \leqslant &\underbrace{\frac{2}{K} \sum^K_{k=1} \left(\decs(\alpha_k) - f(\encs(\alpha_k))\right)^2}_{\ldec(\hat{f})} \nonumber \\ &+  \underbrace{\frac{2}{K} \sum^K_{k=1} \left(f(\encs(\alpha_k)) - f(\mathbf{x}_k)\right)^2}_{\lenc(\hat{f})}.
\end{align}
 The first term represents a proxy for the training error of the encoder regression function. The second term evaluates the decoder regression function at points different from its training data (i.e., $\{\beta_i, f\circ\encs(\beta_i)\}_{i \in \stset}$), thereby measuring the generalization error of the decoder regression function.

Let us start with bounding $\lenc$. Using Lipschitz continuity of  $f(\cdot)$, we obtain:
\begin{align}\label{eq:lenc_final}
     \lenc(\hat{f}) &= \frac{2}{K} \sum^K_{k=1} \left(f(\encs(\alpha_k))- f(x_k)\right)^2 \nonumber \\ 
     &\leqslant \frac{2}{K} \sum^K_{k=1} \left(\nu\cdot|\encs(\alpha_k)- x_k|\right)^2 \nonumber \\
     &= \frac{2\nu^2}{K} \sum^K_{k=1}  (\encs(\alpha_k)- x_k)^2.
\end{align}
Defining $h(z) := \decs(z) - f(\encs(z))$, we obtain:
\begin{align}
    \ldec(\hat{f}) &= \frac{2}{K} \sum^K_{k=1} \left(\decs(\alpha_k) -f(\encs(\alpha_k))\right)^2_2
    \nonumber \\ 
    &\lec{}{=} \frac{2}{K} \sum^K_{k=1} h(\alpha_k)^2
    \nonumber \\
    &\lec{}{\leqslant} \frac{2}{K} \sum^K_{k=1} \norm{h}^2_{\lp{\infty}}
    .
\end{align}
It can be demonstrated that $h \in \hil{2}$ \cite[Lemma~3]{moradicoded}. By applying Sobolev interpolation inequalities \cite{lecun1998gradient}, we can derive the following:
\begin{align}\label{eq:dec_err_bound}
\norm{h}^2_{\lp{\infty}} 
  \leqslant 2\left(\norm{h}^2_{\lp{2}} \cdot \norm{h'}^2_{\lp{2}}\right)^{\frac{1}{2}}.
\end{align}

The solution to \eqref{eq:letcc_dec} is a spline function, known as a \emph{second-order smoothing spline} \cite{wahba1975smoothing, wahba1990spline, wahba2006convergence}. The norm of the error and the norm of its first derivative for the second-order smoothing spline are bounded as follows \cite{ragozin1983error}:
\begin{align}\label{eq:h_1}
    \norm{h}^2_{\lp{2}} {\leqslant}H_0\norm{(f\circ\encs)^{''}}^2_{\lp{2}}\cdot L,
\end{align}
and
\begin{align}\label{eq:hp_1}
\norm{h'}^2_{\lp{2}} \leqslant H_1\norm{(f\circ\encs)^{''}}^2_{\lp{2}} \cdot L^{\frac{1}{2}}(1 + \frac{L}{16})^{\frac{1}{2}},
\end{align}
where $H_0$ and $H_1$ are constants. The term $L$ is given by \begin{align}\label{eq:L_formula}
    L = p_2\left(\frac{\Delta^\stset_\textrm{max}}{\Delta^\stset_\textrm{min}}\right) \cdot \frac{(N - |\stset|) \Delta^\stset_\textrm{max}}{4} \declamb + D \cdot \left(\Delta^\stset_\textrm{max}\right)^4,
    \end{align}
    where $D$ is a constant, $p_2(\cdot)$ is a second-order polynomial with constant and positive coefficients, and $\Delta^\stset_\textrm{max}, \Delta^\stset_\textrm{min}$ are consecutive distance between the decoder mapping points in the non-straggler set $\stset$ as follows:
\begin{gather}
\Delta^\stset_\textrm{max}:=\underset{f \in \left\{0,\dots,|\stset|\right\}}{\max} \left\{\beta_{i_{f+1}}-\beta_{i_f}\right\}, \quad \Delta^\stset_\textrm{min}:=\underset{f \in \left\{1,\dots,|\stset|-1\right\}}{\min} \left\{\beta_{i_{f+1}}-\beta_{i_f}\right\}.
\end{gather}
Combining results in \eqref{eq:h_1} and \eqref{eq:hp_1} with \eqref{eq:dec_err_bound}, we can conclude:
\begin{align}\label{eq:h_norm_infty_bound_2}
    \norm{h}^2_{\lp{\infty}} 
    &\leqslant 2\sqrt{H_0H_1}\norm{(f\circ\encs)^{''}}^2_{\lp{2}}\cdot
    L^{\frac{1}{2}} \left[L^\frac{1}{2}(1 + \frac{L}{16})^\frac{1}{2}\right]^{\frac{1}{2}}.
\end{align}

We model the occurrence of a straggler as a Bernoulli random variable with an expected value of $p$. Consequently, we consider a set of i.i.d. Bernoulli random variables $\{B_1, \dots, B_N\}$, where each variable represents the state of a server. Specifically, $B_i = 1$ indicates that server $i$ is a straggler, while $B_i = 0$ means that the master node has received the results from server $i$. 

Next, we define the random variable $R_{\stset, N}$ as the length of the longest subsequence of ones in the sequence $\{B_1, \dots, B_N\}$. This random variable is commonly referred to in the literature as the \emph{longest head run} \cite{ni2019asymptotic,sinha2009distribution,bateman1948power,philippou1986successes,gordon1986extreme}.
 Therefore, we have:
\begin{align}\label{eq:longest_straggler_ineq_2}
     \Delta^\stset_\textrm{max} \leqslant (R_{\stset, N}+1) \cdot \Delta_\textrm{max},\quad \frac{\Delta^\stset_\textrm{max}}{\Delta^\stset_\textrm{min}} \lec{(a)}{\leqslant}(R_{\stset, N}+1) \cdot \frac{\Delta_\textrm{max}}{\Delta_\textrm{min}}\lec{}{\leqslant}(R_{\stset, N}+1)B,
\end{align}
where (a) is due to $\Delta^\stset_\textrm{max} \geqslant \Delta_\textrm{min}$. Note that since $\Delta_\textrm{min} \leqslant \frac{2}{N}$ and $\frac{\Delta^\textrm{max}}{\Delta^\textrm{min}} \leqslant B$, there exists a constant $J$ such that $\Delta^\textrm{max} \leqslant \frac{J}{N}$.

Using \eqref{eq:longest_straggler_ineq_2} and definition of $L$, one can conclude:
\begin{align}\label{eq:L_ineq}
    L &\lec{}{\leqslant} p_3\left(R_{\stset, N} + 1\right) \cdot \left(N - |\stset|\right)\Delta_\textrm{max}\declamb + D\left(R_{\stset, N}+1\right)^4 \left(\Delta_\textrm{max}\right)^4\nonumber \\
    &\lec{(a)}{\leqslant} p_3\left(R_{\stset, N} + 1\right) \cdot J\frac{\left(N - |\stset|\right)}{N}\declamb + D\left(R_{\stset, N}+1\right)^4 N^{-4}
    \nonumber \\
    &\lec{(b)}{\leqslant} J\frac{p_3\left(R_{\stset, N} + 1\right)}{N^{4}}  + D\frac{\left(R_{\stset, N}+1\right)^4}{N^{4}}
    \nonumber \\
    &\lec{}{\leqslant} \tilde{J}\frac{\left(R_{\stset, N}+1\right)^4}{N^{4}},
\end{align}
where $p_3(R_{\stset, N} + 1):=p_2\left(B(R_{\stset, N}+1)\right) \cdot \frac{(R_{\stset, N}+1)}{4}$ is a degree-3 polynomial of $(R_{\stset, N}+1)$ with constant positive coefficients. Step (a) follows from $\Delta_\textrm{max} \leqslant \frac{J}{N}$, and (b) is due to $\declamb \leqslant N^{-4}$ and $N - |\stset| \leqslant N$. Here, $\tilde{J}$ is defined as $\tilde{J}:=J \cdot p_3(1) + D$.

\begin{lemma}\label{lem:longest_run_2}
    With the same assumption of Lemma~\ref{lem:longest_run}, for every $\delta \in (0, 1)$, there exist constants $C, n_0>0$ such that for every $N>n_0$:
    \begin{align}
        \Pr\left(R_{\stset, N} > C\left(\log_{\frac{1}{p}}(qN) + \sqrt{\frac{1}{\delta}}\right)\right) < \delta
    \end{align}
\end{lemma}
\begin{proof}
    Let us define $\mu_1 := \frac{\theta^{\frac{1}{2}}}{2 \pi e^\theta\left(1-e^{-\theta}\right)^2} + \frac{\gamma}{\ln \left(\frac{1}{p}\right)}-\frac{1}{2} + 1$ and $\mu_2 := \frac{2(1.1+0.7 \theta) \theta^{\frac{1}{2}}}{2 \pi e^\theta\left(1-e^{-\theta}\right)^3} + \frac{\pi^2}{6\left(\ln \left(\frac{1}{p}\right)\right)^2}+\frac{1}{12} + 1$ for the sake of simplicity. Therefore, there exists an integer $n_0 > 0$ such that for all $N > n_0$, the $o(1)$ asymptotic error terms are strictly bounded by $1$. This gives us:
\begin{align}
    \frac{\gamma}{\ln (1/p)} - \frac{1}{2} + g_1(\theta) + o(1) &\leqslant \mu_1, \\
    \frac{\pi^2}{6(\ln (1/p))^2} + \frac{1}{12} + g_2(\theta) + o(1) &\leqslant \mu_2.
\end{align}

    Thus, we have:
    \begin{align}
        \Pr\left(R_{\stset, N} \geqslant \log_{\frac{1}{p}}(qN) + \mu_1 + c\sqrt{\mu_2} \right) &\lec{(a)}{\leqslant}  \Pr\left(R_{\stset, N} \geqslant  \mathrm{E}[R_{\stset, N}] + c\sqrt{\mu_2} \right) \nonumber \\
        &\lec{(b)}{\leqslant}  \Pr\left(R_{\stset, N} \geqslant  \mathrm{E}[R_{\stset, N}] + c\cdot \mathrm{Var}(R_{\stset, N})\right) \nonumber \\
        &\lec{(c)}{\leqslant}  \frac{1}{1 + c^2},
    \end{align}
    where (a) follows from $\mathrm{E}[R_{\stset, N}] \leqslant \log_{\frac{1}{p}}(qN) + \mu_1$, (b) is due to $\mathrm{Var}[R_{\stset, N}] \leqslant \mu_2$, and (c) follows from the one-sided Chebyshev inequality. Thus, selecting $c = \sqrt{\frac{1}{\delta}-1}$ completes the proof.
\end{proof}
Using Lemma~\ref{lem:longest_run_2} and \eqref{eq:L_ineq}, one can show that:
\begin{align}\label{eq:L_ineq_prob}
    \Pr\left(L \leqslant \tilde{J}\frac{\left(
    \log_{\frac{1}{p}}(qN) + \sqrt{\frac{1}{\delta}}
    \right)^4}{N^{4}}\right) \geqslant \Pr\left(L \leqslant \tilde{J}\frac{\left(
    \log_{\frac{1}{p}}(qN) + \sqrt{\frac{1}{\delta}-1}
    \right)^4}{N^{4}}\right) \geqslant 1- \delta.
\end{align}
Therefore, with the probability of at least $1- \delta$, we have:
\begin{align}\label{eq:ldec_proof_final}
    \ldec(\hat{f}) &\lec{}{\leqslant} \norm{(f\circ\encs)^{(2)}}^2_{\lp{2}} \left[2H_0^\frac{1}{2}H_1^\frac{1}{2}\Tilde{J}^{\frac{3}{4}}\cdot \left(\frac{
    \log_{\frac{1}{p}}(qN) + \sqrt{\frac{1}{\delta}}
    }{N}\right)^3\cdot \left(1 + \frac{\Tilde{J}}{16}\left(\frac{
    \log_{\frac{1}{p}}(qN) + \sqrt{\frac{1}{\delta}}
    }{N}\right)^4\right)^\frac{1}{4}
    \right] \nonumber \\
    &\lec{(a)}{\leqslant} 2H_0^\frac{1}{2}H_1^\frac{1}{2}\Tilde{J}^{\frac{3}{4}}\left(1+\frac{\Tilde{J}}{16}\right)^\frac{1}{4}\cdot \norm{(f\circ\encs)^{(2)}}^2_{\lp{2}} \left(\frac{
    \log_{\frac{1}{p}}(qN) + \sqrt{\frac{1}{\delta}}
    }{N}\right)^3,
\end{align}
where (a) follows from $\log_{\frac{1}{p}}(qN) + \sqrt{\frac{1}{\delta}} \leqslant N$. Combining \eqref{eq:ldec_proof_final} and \eqref{eq:lenc_final} completes the proof.

\subsection{Proof of Theorem~\ref{th:berrut}}\label{sec:app_proof_th_berrut}

For the $\texttt{BACC}$ scheme, we define $h(z) := \decs(z) - f(\encs(z))$. The following upper bound for $\norm{h}_{\lp{\infty}}$ is given in \cite{floater2007barycentric}:
\begin{align}
    \norm{h}_{\lp{\infty}} \leqslant \Delta^\stset_\textrm{max}(1+\mu) \cdot \sum^2_{i=1} \norm{(f \circ \encs)^{(i)}}_{\lp{\infty}},
\end{align}
where $\mu := \max_{1 \leq i \leq |\stset|-2} \min \left\{\frac{\beta_{j_{i+1}}-\beta_{j_i}}{\beta_{j_i}-\beta_{j_{i-1}}}, \frac{\beta_{j_{i+1}}-\beta_{j_i}}{\beta_{j_{i+2}}-\beta_{j_{i+1}}}\right\}$ and $\stset := \{\beta_{j_1}, \dots, \beta_{j_{|\stset|}}\}$. Applying Max-min inequality, we have:
\begin{align}
    \mu &\leqslant \min \left\{ \max_{1 \leq i \leq |\stset|-2} \frac{\beta_{j_{i+1}}-\beta_{j_i}}{\beta_{j_i}-\beta_{j_{i-1}}}, \max_{1 \leq i \leq |\stset|-2} \frac{\beta_{j_{i+1}}-\beta_{j_i}}{\beta_{j_{i+2}}-\beta_{j_{i+1}}}\right\} \nonumber \\ 
    &\leqslant \frac{\Delta^\stset_\textrm{max}}{\Delta^\stset_\textrm{min}}.
\end{align}
Therefore, using the inequalities from \eqref{eq:longest_straggler_ineq_2}, we obtain:
\begin{align}\label{eq:berrut_bound_final}
\norm{h}^2_{\lp{\infty}}&\lec{}{\leqslant} \left(\Delta^\stset_\textrm{max}\left(1+\frac{\Delta^\stset_\textrm{max}}{\Delta^\stset_\textrm{min}}\right)\right)^2   \cdot\left(\sum^2_{i=1} \norm{(f \circ \encs)^{(i)}}_{\lp{\infty}}\right)^2 \nonumber \\
 &\lec{}{\leqslant} \left(\Delta_\textrm{max}\right)^2 \cdot \left(R_{\stset, N}+1\right)^2 \cdot \left(1+ B\left(R_{\stset, N}+1\right)\right)^2 \cdot\left(\sum^2_{i=1} \norm{(f \circ \encs)^{(i)}}_{\lp{\infty}}\right)^2
 \nonumber \\
 &\lec{(a)}{\leqslant}  \frac{J^2\left(R_{\stset, N}+1\right)^2}{N^2} \cdot \left(1+ B\left(R_{\stset, N}+1\right)\right)^2 \cdot\left(\sum^2_{i=1} \norm{(f \circ \encs)^{(i)}}_{\lp{\infty}}\right)^2
  \nonumber \\
 &\lec{}{\leqslant}  2J^2B\cdot \frac{\left(R_{\stset, N}+1\right)^4}{N^2}  \cdot\left(\sum^2_{i=1} \norm{(f \circ \encs)^{(i)}}_{\lp{\infty}}\right)^2,
\end{align}
where (a) is based on the inequality $\Delta_\textrm{max} \leqslant \frac{J}{N}$.
Finally, applying Lemma~\ref{lem:longest_run_2} to \eqref{eq:berrut_bound_final} completes the proof.

\subsection{Proof of Corollary~\ref{cor:rate}}\label{sec:app_proof_cor_rate}
Since $f\circ \encs(\cdot)$ does not depend on $N$, the convergence rate of the $\texttt{BACC}$'s average approximation error follows directly from Theorem~\ref{th:berrut}. For the $\nprcc$ scheme, consider a natural spline function fitted to the data $\{\alpha_k, x_k\}_{k=1}^K$. Additionally, let $\encs^*(\cdot)$ denote the minimizer of the upper bound in \eqref{eq:th_letcc}. Therefore, we have:
\begin{align}
    \mathcal{L}_\stset(\hat{f}) &\leqslant \frac{2\nu^2}{K} \sum_{k=1}^K (\encs^*(\alpha_k) - x_k)^2 + C_l\frac{\left(\log_{\frac{1}{p}}(qN) + \sqrt{\frac{1}{\delta}}\right)^3}{N^3} \cdot \norm{(f \circ \encs^*)''}^2_{\lp{2}} 
    \nonumber \\
    &\leqslant \frac{2\nu^2}{K} \sum_{k=1}^K (\widetilde{\encs}(\alpha_k) - x_k)^2 + C_l\frac{\left(\log_{\frac{1}{p}}(qN) + \sqrt{\frac{1}{\delta}}\right)^3}{N^3} \cdot \norm{(f \circ \widetilde{\encs})''}^2_{\lp{2}}  \nonumber \\
    &\lec{(a)}{\leqslant} C_l\frac{\left(\log_{\frac{1}{p}}(qN) + \sqrt{\frac{1}{\delta}}\right)^3}{N^3} \cdot \norm{(f \circ \widetilde{\encs})''}^2_{\lp{2}},
\end{align}
where (a) holds because the natural spline overfits the input data, i.e., $\widetilde{\encs}(\alpha_k) = x_k$ for all $k \in [K]$. Since $f \circ \widetilde{\encs}(\cdot)$ is independent of $N$, the proof is complete.

\subsection{Proof of Corollary~\ref{cor:cheb}}\label{sec:app_proof_cor_cheb}
Based on the definition of Chebyshev points, we have the following expressions for the minimum and maximum consecutive distance of these points:
\begin{align}
    \Delta_\textrm{max} = 
    \begin{cases}
        \cos\left(\frac{\pi\left(\lfloor\frac{N}{2}\rfloor - 1\right)}{N-1}\right) - \cos\left(\frac{\pi\left(\lceil\frac{N}{2}\rceil - 1\right)}{N-1}\right)  & \text{if $N$ is odd}, \\
          \cos\left(\frac{\pi\left(\frac{N}{2}-1\right)}{N-1}\right) - \cos\left(\frac{\pi\frac{N}{2}}{N-1}\right) & \text{if $N$ is even},
    \end{cases},\quad
    \Delta_\textrm{min} = 1 - \cos\left(\frac{\pi}{N-1}\right).
\end{align}
Thus, we have:
\begin{align}
    \Delta_\textrm{max} &= 
    \begin{cases}
    2 \sin\left(\frac{\pi(N-2)}{2(N-1)}\right) \sin\left(\frac{\pi}{2(N-1)}\right)  &\text{if $N$ is odd}, \\
    2 \sin\left(\frac{\pi}{2(N-1)}\right)  &\text{if $N$ is even}
    \end{cases}
     = \mathcal{O}\left(\frac{1}{N}\right), \\
    \Delta_\textrm{min} &= 2 \sin^2\left(\frac{\pi}{2(N-1)}\right) = \mathcal{O}\left(\frac{1}{N^2}\right).
\end{align}
As a result $\frac{\Delta_\textrm{max}}{\Delta_\textrm{min}} \leqslant \frac{1}{\sin\left(\frac{\pi}{2(N-1)}\right)} = \mathcal{O}\left(N\right)$. 

Let us start with $\nprcc$ scheme. Substituting in \eqref{eq:L_formula}, for sufficiently large $N$ and $\declamb \leqslant N^{-6}$, we have:
\begin{align}
    L &= p_2\left(\frac{\Delta^\stset_\textrm{max}}{\Delta^\stset_\textrm{min}}\right) \cdot \frac{(N - |\stset|) \Delta^\stset_\textrm{max}}{4} \declamb + D \cdot \left(\Delta^\stset_\textrm{max}\right)^4 \nonumber \\
    &\lec{(a)}{\leqslant} p_2\left(J_1N\cdot(R_{\stset, N} + 1)\right) \cdot \frac{J_2(R_{\stset, N} + 1)(N - |\stset|)}{4N} \declamb + D \left(\frac{J_2(R_{\stset, N} + 1)}{N}\right)^4
    \nonumber \\
    &\lec{}{\leqslant} \Tilde{C}_1(R_{\stset, N} + 1)^3N^2\declamb + \Tilde{C}_2 (R_{\stset, N} + 1)^4 N^{-4},
    \nonumber \\
    &\lec{(b)}{\leqslant} \Tilde{C}\frac{(R_{\stset, N} + 1)^4}{N^{-4}},
\end{align}
where (a) follows from $\frac{\Delta_\textrm{max}}{\Delta_\textrm{min}} \leqslant J_1N$ and $\Delta_\textrm{max} \leqslant \frac{J_2}{N}$, and (b) is due to $\declamb \leqslant N^{-6}$. Following similar steps as those used in the proof of Theorem~\ref{th:letcc} completes the proof.

For the $\texttt{BACC}$ scheme, the steps done in \eqref{eq:berrut_bound_final} will results in an upper bound that does not converge to zero as $N$ increases because:
$$
\left(\Delta^\stset_\textrm{max}\left(1+\frac{\Delta^\stset_\textrm{max}}{\Delta^\stset_\textrm{min}}\right)\right)^2 \leqslant \left(\frac{J_2}{N}(R_{\stset, N}+1)\left(1 + J_1(R_{\stset, N}+1)N\right)\right)^2 \leqslant \mathcal{O}\left(R_{\stset, N}^4\right) \approx \mathcal{O}\left(\left(\log_{\frac{1}{p}}(qN)\right)^4\right).
$$
However, \cite{jahani2022berrut} demonstrated that if the maximum number of consecutive stragglers is $s$, the following holds:
\begin{align}\label{eq:bound_mu}
    \mu &\leqslant \min \left\{ \max_{1 \leq i \leq |\stset|-2} \frac{\beta_{j_{i+1}}-\beta_{j_i}}{\beta_{j_i}-\beta_{j_{i-1}}}, \max_{1 \leq i \leq |\stset|-2} \frac{\beta_{j_{i+1}}-\beta_{j_i}}{\beta_{j_{i+2}}-\beta_{j_{i+1}}}\right\} \leqslant \frac{(s+1)(s+3)\pi^2}{4}.
\end{align}
Therefore, we have:
\begin{align}
    \norm{h}^2_{\lp{\infty}}  &\leqslant \left(\Delta^\stset_\textrm{max}\right)^2(1+\mu)^2 \cdot \left(\sum^2_{i=1} \norm{(f \circ \encs)^{(i)}}_{\lp{\infty}}\right)^2 \nonumber \\ &\lec{(a)}{\leqslant} \left(\frac{J_2(R_{\stset, N}+1)}{N}\right)^2(1+\frac{((R_{\stset, N}+1)+1)((R_{\stset, N}+1)+3)\pi^2}{4})^2 \cdot \left(\sum^2_{i=1} \norm{(f \circ \encs)^{(i)}}_{\lp{\infty}}\right)^2 \nonumber \\
    &\lec{}{\leqslant} \tilde{C}\frac{(R_{\stset, N}+1)^4}{N^2} \cdot \left(\sum^2_{i=1} \norm{(f \circ \encs)^{(i)}}_{\lp{\infty}}\right)^2,
\end{align}
where (a) is follows from \eqref{eq:bound_mu} and $\Delta^\stset_\textrm{max} \leqslant \frac{J_2}{N}(R_{\stset, N}+1)$. Using Lemma~\ref{lem:longest_run_2} completes the proof.
\end{document}